\documentclass[letterpaper, 10 pt, conference]{ieeeconf}  

\IEEEoverridecommandlockouts                              

\usepackage{amsmath}
\usepackage{amssymb}
\usepackage{tikz}
\usepackage{pgfplots}
\usepackage{graphicx}
\usepackage{setspace}
\usepackage{caption}
\usepackage{subcaption}
\usepackage{xcolor}
\allowdisplaybreaks

\newcommand{\rs}[1]{\rho^{#1}(\boldsymbol{x},t)}
\newcommand{\rsoo}[1]{\rho^{#1}(\boldsymbol{x},0)}
\newcommand{\rsnt}[1]{\rho^{#1}(\boldsymbol{x})}

\newcommand{\rsntt}[1]{\rho^{#1}(\boldsymbol{x}(t))}

\newcommand{\rss}[2]{\rho^{#1}(\boldsymbol{x},{#2})}
\newcommand{\rso}[1]{\rho^{#1}_{opt}}

\newcommand{\bd}[1]{\boldsymbol{#1}}

\newtheorem{definition}{Definition} 
\newtheorem{theorem}{Theorem} 
\newtheorem{assumption}{Assumption} 
\newtheorem{problem}{Problem} 
\newtheorem{remark}{Remark}
\newtheorem{lemma}{Lemma}

\newtheorem{example}{Example}

\title{\LARGE \bf
Event-triggered Feedback Control for Signal Temporal Logic Tasks
}

\author{Lars Lindemann$^{a,*}$, Dipankar Maity$^{b,*}$, John S. Baras$^{a,b}$, and Dimos V. Dimarogonas$^a$
\thanks{$^*$ The authors contributed equally.}
\thanks{The work of L. Lindemann and D. V. Dimarogonas was supported in part by the Swedish Research Council, the European Research Council, the Swedish Foundation for Strategic Research,  the EU H2020 Co4Robots project, the SRA ICT TNG project STaRT, and the Knut and Alice Wallenberg Foundation. The work of D. Maity and J. S. Baras was supported in part by DARPA through ARO grant W911NF1410384, ONR grant N00014-17-1-2622, the Knut and Alice Wallenberg Foundation, the Swedish Foundation for Strategic Research, and the Swedish Research Council.}
\thanks{$^a$The authors are with the Department of Automatic Control, School of Electrical Engineering and Computer Science, Royal Institute of Technology (KTH), 100 44 Stockholm, Sweden. {\tt\small llindem@kth.se (L. Lindemann), dimos@kth.se (D.V. Dimarogonas)}}%
\thanks{$^b$ The authors are with the Department of Electrical and Computer Engineering, and the Institute for Systems Research at University of Maryland, USA. {\tt\small dmaity@umd.edu (D. Maity), baras@umd.edu (J.S. Baras)}}
}

\begin{document}

\maketitle
\thispagestyle{empty}
\pagestyle{empty}

\begin{abstract}

A framework for the event-triggered control synthesis under signal temporal logic (STL) tasks is proposed. In our previous work, a continuous-time feedback control law was designed, using the prescribed performance control technique, to satisfy STL tasks. We replace this continuous-time feedback control law by an event-triggered controller. The event-triggering mechanism is based on a maximum triggering interval and on a norm bound on the difference between the value of the current state and the value of the state at the last triggering instance. Simulations of a multi-agent system  quantitatively show the efficacy of using an event-triggered controller to reduce communication and computation efforts. 
\end{abstract}

\section{Introduction}
\label{sec:introduction}

Robot motion planning has traditionally been a challenging problem to the control community. Initially, the studies were primarily directed towards optimal navigation from an initial to a goal position while avoiding obstacles \cite{karaman2011sampling}. The focus, however, shifted over the years to integrate complex high-level task descriptions with the low-level dynamics of the robots. Consequently, temporal logics \cite{baier} were brought into the domain of motion planning to formally express and systematically address such  complex behaviors in a generic way \cite{belta2007symbolic,fainekos2009temporal}. Temporal logics have a rich expressivity; however, integrating the temporal logic descriptions with the dynamics of the robot is far from trivial.  The derived methods are highly relying on automata theory and an appropriate abstraction of the dynamics. As a result, complexity problems arise as the size of the automata grow exponentially with the `size of the task' and with the complexity of the dynamics. Efficient techniques have been proposed in the context of temporal logic-based design. Nonetheless, integration of temporal logic tasks with a `finer' abstraction of the dynamics is still computationally expensive. Moreover, temporal logic formulae expressing real-time constraints such as in signal temporal logic (STL) \cite{maler1} are even more intricate to handle.

In this work, we take a different approach of integrating high level STL tasks with the dynamics of the robot. Instead of following the automata-based approach, we aim for a feedback control law that maximizes a robustness metric associated with the temporal logic formula. STL was introduced in \cite{maler1}, while space robustness \cite{donze2} is the aforementioned robustness metric for STL. In our previous work \cite{lindemann2017prescribed}, we leveraged ideas from prescribed performance control \cite{bechlioulis2014low} to derive a feedback control law that satisfies the STL task under consideration. Prescribed performance control essentially allows to impose a transient behavior to the robustness metric, that, if properly designed, results in a satisfaction of the STL task. To implement such a feedback control law, we need continuous transmission of the measurements from the sensors to the controllers, and this can be a bottleneck in implementations. Mobile robots, e.g., operating in uncertain environments, have limited energy and bandwidth to transmit continuous measurements to the controllers. To alleviate this problem, we delve into synthesizing an event-triggered feedback control law in this paper that will ensure the satisfaction of the STL task. Event-triggered control is an approach that has gained increasing attention recently \cite{tabuada2007event}. Extensions to multi-agent systems have appeared, e.g., in \cite{dimarogonas2012distributed}. An overview of this topic in the setup of hybrid systems can be found in  \cite{postoyan2015framework}. A first attempt in combining event-triggered control and temporal logic-based specifications is presented in \cite{maity2018event}. However, the synthesis of the control for the satisfaction of the temporal-logic formula is based on the automata theory and discretization methods, and hence it needs to deal with the high computational complexity issues associated with such methods. 

The main contribution of this paper is an event-triggering feedback control law for dynamical systems under STL task specifications. This event-triggered feedback control law is robust with respect to noise and the task satisfaction, while, at the same time, avoiding discretizations. We emphasize that the major difference compared with our previous work in \cite{lindemann2017prescribed} is the event-based nature of the approach taken here due to which a drastic reduction in communication is observed.

The rest of the paper is organized as follows: notations and preliminaries are provided in Section \ref{sec:background}, the formal problem definition is given in Section \ref{sec:problem_formulation}, while Section \ref{sec:strategy} studies the control synthesis problem. Simulations are performed in Section \ref{sec:simulations}, and finally we conclude the paper in Section \ref{sec:conclusion}.
\section{Notation and Preliminaries}
\label{sec:background}

Scalars are denoted by lowercase, non-bold letters $x$ and column vectors are lowercase, bold letters $\boldsymbol{x}$. True and false are denoted by $\top$ and $\bot$; $\mathbb{R}^n$ is the $n$-dimensional vector space over the real numbers $\mathbb{R}$. The natural, non-negative, and positive real numbers are $\mathbb{N}$, $\mathbb{R}_{\ge0}$, and $\mathbb{R}_{>0}$, respectively. Let $\|\cdot\|$ and $\|\cdot\|_\infty$ denote the $L_2$-norm and $L_\infty$-norm, respectively, and set $B_{\delta}(\bd x'):=\{\boldsymbol{x}\in\mathbb{R}^n|\|\boldsymbol{x}-\boldsymbol{x}'\|_\infty {\le} \delta\}$.

Let $\boldsymbol{x}\in\mathbb{R}^n$, $\boldsymbol{u}\in\mathbb{R}^m$, and $\boldsymbol{w}\in \mathcal{W}\subset\mathbb{R}^n$, where $\mathcal{W}$ is bounded,  be the state, input, and additive noise of a  system 
\begin{align}\label{system_noise}
\dot{\boldsymbol{x}}(t)&=f(\boldsymbol{x}(t))+g(\boldsymbol{x}(t))\boldsymbol{u}+\boldsymbol{w}(t)
\end{align}
where $f$ is unknown apart from a regularity assumption. 

\begin{assumption}\label{assumption:1}
	The functions $f:\mathbb{R}^n\to\mathbb{R}^n$ and $g:\mathbb{R}^n\to\mathbb{R}^{n\times m}$ are  locally Lipschitz continuous, the function $\boldsymbol{w}:\mathbb{R}_{\ge 0}\to\mathbb{R}^n$ is piecewise continuous, and $g(\boldsymbol{x}){g(\boldsymbol{x})}^T$ is positive definite for all $\boldsymbol{x}\in \mathbb{R}^n$, i.e., $\exists \lambda_{\min}\in\mathbb{R}_{>0}$ s.t. $\lambda_{\min}\|\boldsymbol{z}\|^2\le \boldsymbol{z}^Tg(\boldsymbol{x})g^T(\boldsymbol{x}) \boldsymbol{z}$ for all $\boldsymbol{x},\boldsymbol{z}\in \mathbb{R}^n$.
\end{assumption}

We emphasize again that $f$ is unknown so that \eqref{system_noise} is \emph{not} feedback equivalent to $\dot{\boldsymbol{x}}(t)=\boldsymbol{u}(t)+\boldsymbol{w}(t)$.

\subsection{Signal Temporal Logic (STL)}
Signal temporal logic (STL) consists of predicates $\mu$ that are, {for some $\boldsymbol{\zeta}\in\mathbb{R}^n$}, evaluated by a {continuously differentiable} predicate function $h:\mathbb{R}^n\to\mathbb{R}$  as
\begin{align*}
\mu:=
 \begin{cases} 
 \top \text{ if } h(\boldsymbol{\zeta})\ge 0\\
 \bot \text{ if } h(\boldsymbol{\zeta})< 0.
 \end{cases}
 \end{align*}
 
 The STL syntax is then given by
\begin{align*}
\phi \; ::= \; \top \; | \; \mu \; | \; \neg \phi \; | \; \phi_1 \wedge \phi_2 \; | \; F_{[a,b]}\phi\;|\; G_{[a,b]}\phi
\end{align*}
where $\phi_1$, $\phi_2$ are STL formulas and where $a\in\mathbb{R}_{\ge 0}$ and $b\in \mathbb{R}_{\ge 0}\cup \infty$ with $a\le b$.  The satisfaction relation $(\boldsymbol{x},t)\models \phi$ denotes that the signal $\boldsymbol{x}:\mathbb{R}_{\ge 0}\to\mathbb{R}^n$ satisfies  $\phi$ at time $t$. The STL semantics \cite[Definition 1]{maler1} are recursively defined as: $(\boldsymbol{x},t) \models \mu \Leftrightarrow h(\boldsymbol{x}(t))\ge 0$, $(\boldsymbol{x},t) \models \neg\phi \Leftrightarrow \neg((\boldsymbol{x},t) \models \phi)$, $(\boldsymbol{x},t) \models \phi_1 \wedge \phi_2 \Leftrightarrow (\boldsymbol{x},t) \models \phi_1 \wedge (\boldsymbol{x},t) \models \phi_2$, $(\boldsymbol{x},t) \models F_{[a,b]}\phi \Leftrightarrow \exists t_1 \in[t+a,t+b] \text{ s.t. }(\boldsymbol{x},t_1)\models \phi$, and $(\boldsymbol{x},t) \models G_{[a,b]}\phi \Leftrightarrow \forall t_1 \in[t+a,t+b] \text{, }(\boldsymbol{x},t_1)\models \phi$. Disjunction and until operators are not considered in this paper. Space robustness \cite{donze2} are robust semantics for STL, which are given in Definition \ref{def:2} and denoted by $\rs{\phi}$.  Space robustness determines how robustly a signal $\boldsymbol{x}$ satisfies the formula $\phi$ at time $t$ and it holds that $(\boldsymbol{x},t)\models \phi$ if $\rs{\phi}>0$ \cite{fainekos2009robustness}.
\begin{definition}\cite[Definition 3]{donze2} {The semantics of space robustness are recursively given by:}
\begin{align*}
\rs{\mu}& := h(\boldsymbol{x}(t))\\
\rs{\neg\phi} &:= 	-\rs{\phi}\\
\rs{\phi_1 \wedge \phi_2} &:= 	\min\big(\rs{\phi_1},\rs{\phi_2}\big)\\
\rs{F_{[a,b]} \phi} &:= \underset{t_1\in[t+a,t+b]}{\max}\rss{\phi}{t_1}\\
\rs{G_{[a,b]} \phi} &:= \underset{t_1\in[t+a,t+b]}{\min}\rss{\phi}{t_1}.
\end{align*}
\label{def:2}
\end{definition}

By a slight change of notation, let $\rho^\phi(\boldsymbol{x}(t)):=\rs{\phi}$ if $t$ is not explicitly contained in $\rs{\phi}$, i.e., $t$ is contained in $\rho^\phi$ only through the composition of $\rho^\phi$ with the signal $\boldsymbol{x}$. For instance, $\rho^\mu(\boldsymbol{x}(t)) := h(\boldsymbol{x}(t))$ since $h(\boldsymbol{x}(t))$ is the composition of $h$ with $\boldsymbol{x}$. However, $t$ is explicitly contained in $\rs{\phi}$ for temporal operators (eventually or always).

%
%

\subsection{Event-Triggered Control} \label{sec:event}

In contrast to continuous feedback control, event-triggered feedback control requires the state measurements intermittently while ensuring stability and sometimes even performance arbitrarily close to the continuous feedback controller \cite{maity2018event}. These controllers rely on an event generator that decides on the instances when the state measurements are sent to the controller. Thereby, the communication between the sensors and the controller can be reduced. Thus, in event-triggered control, the controller obtains new state information only at certain discrete time instances denoted by $t_1,t_2,\cdots,t_i,\cdots$. 
%

For any given event-triggered controller, it needs to be guaranteed that Zeno behavior is avoided, i.e., the case of infinite switching in finite time. This will be explicitly shown in our design by guaranteeing a strictly positive minimum inter-triggering time, which means that $t_{i+1}-t_i$ for each $i\ge 1$ is positive and lower bounded.

  
\subsection{Prescribed Performance Control (PPC)}
Prescribed performance control (PPC) \cite{bechlioulis2014low} constrains a generic error $\boldsymbol{e}:\mathbb{R}_{\ge 0}\to\mathbb{R}^n$ to a user-designed funnel. For instance, consider $\boldsymbol{e}(t):=\begin{bmatrix}e_1(t) & \hdots e_n(t) \end{bmatrix}^T:=\boldsymbol{x}(t)-\boldsymbol{x}_d(t)$, where $\boldsymbol{x}_d$ is a desired trajectory. In order to prescribe transient and steady-state behavior to this error, let us define the performance function $\gamma$ in Definition \ref{def:p} as well as the transformation function $S$ in Definition \ref{def:SS}.
\begin{definition}\label{def:p}\cite{bechlioulis2014low} 
	A performance function $\gamma:\mathbb{R}_{\ge 0}\to\mathbb{R}_{> 0}$ is a continuously differentiable, bounded, positive, and non-increasing function given by
	$\gamma(t):=(\gamma_0-\gamma_\infty)\exp(-lt)+\gamma_\infty$ where $\gamma_0, \gamma_\infty \in \mathbb{R}_{>0}$ with $\gamma_0\ge \gamma_\infty$ and $ l \in \mathbb{R}_{\ge 0}$.
\end{definition}
\begin{definition}\label{def:SS}\cite{bechlioulis2014low} 
A transformation function $S:(-1,M)\to\mathbb{R}$ with $M\in[0,1]$ is a smooth  and strictly increasing function. Let $S(\xi):=\ln\left(-\frac{\xi+1}{\xi-M}\right)$.
\end{definition}

Now let $\gamma_i$ be a performance function in the sense of Definition \ref{def:p}. The task is to synthesize a continuous feedback control law such that each error $e_i$ satisfies
\begin{align}
-\gamma_i(t)<e_i(t)<M\gamma_i(t) \;\;\;   \forall t\in \mathbb{R}_{\ge 0}, \forall i\in \{1,\hdots,n\}\label{eq:constrained_funnel}
\end{align}
given that $-\gamma_i(0)<e_i(0)<M\gamma_i(0)$; $\gamma_i$ is a design parameter by which transient and steady-state behavior of $e_i$ can be prescribed. Note also that \eqref{eq:constrained_funnel} is a constrained control problem with $n$ constraints subject to the dynamics in \eqref{system_noise}. Next, define the normalized error $\xi_i:=\frac{e_i}{\gamma_i}$. Dividing \eqref{eq:constrained_funnel} by $\gamma_i$ and applying the transformation function $S$ results in an unconstrained control problem  $-\infty<S\big(\xi_i(t)\big)<\infty$ with the transformed error $\epsilon_i:=S\big(\xi_i\big)$. If $\epsilon_i(t)$ is bounded for all $t\ge \mathbb{R}_{\ge 0}$, then $e_i$ satisfies \eqref{eq:constrained_funnel}.

\section{Problem Definition}
\label{sec:problem_formulation}
In this paper, we consider the following STL fragment
\begin{subequations}\label{eq:subclass}
\begin{align}
\psi \; &::= \; \top \; | \; \mu \; | \; \neg \mu \; | \; \psi_1 \wedge \psi_2\label{eq:psi_class}\\
\phi \; &::= \;  G_{[a,b]}\psi \; | \; F_{[a,b]} \psi \label{eq:phi_class}\\
\theta^{s_1} &::= \bigwedge \limits_{k=1}^{K} \phi_k \text{ with } b_k\le a_{k+1} \label{eq:theta1_class} \\
\theta^{s_2} &::=  \tilde{\phi}_1\label{eq:theta2_class}\\
\theta &::= \theta^{s_1} \; |\; \theta^{s_2}.\label{eq:theta_class}
\end{align}
\end{subequations}
where $\psi$ in \eqref{eq:phi_class} and $\psi_1, \psi_2$ in \eqref{eq:psi_class}  are formulas of class $\psi$ given in \eqref{eq:psi_class}. Formulas $\phi_k$ with $k\in\{1,\hdots,K\}$ in \eqref{eq:theta1_class} are formulas of class $\phi$ given in \eqref{eq:phi_class} with time intervals $[a_k,b_k]$, whereas $\tilde{\phi}_1$ in \eqref{eq:theta2_class} follows from the recursive definition $\tilde{\phi}_k:=F_{[c_k,d_k]}(\psi_k\wedge\tilde{\phi}_{k+1}) \text{ for all } k\in\{1,\hdots K-1\}$ and $\tilde{\phi}_K:=F_{[c_K,d_K]}\psi_K$ with $c_k,d_k\in\mathbb{R}_{\ge 0}$ and $c_k\le d_k$ for all $k\in\{1,\hdots K\}$. We refer to $\psi$ as \emph{non-temporal formulas} and use $\rsntt{\psi} :=\rs{\psi}$ due to the previous discussion. In contrast, $\phi$ and $\theta$ are referred to as \emph{temporal formulas}. 

{The control strategy that will be introduced in Section \ref{sec:strategy} requires three additional assumptions explained next. First, for conjunctions of \emph{non-temporal formulas} of class $\psi$ given in \eqref{eq:psi_class}, e.g., $\psi:=\psi_1\wedge\psi_2$, we approximate the robust semantics in Definition~\ref{def:2}, e.g., $\rho^{\psi_1\wedge\psi_2}$,  by a smooth function.
\begin{assumption}\label{assumption:2}
The robust semantics for a conjunction of \emph{non-temporal formulas} of class $\psi$ given in \eqref{eq:psi_class}, i.e., $\rsntt{\psi_1 \wedge \psi_2}$, are approximated by a smooth function as 
\begin{align*}
\rsntt{\psi_1 \wedge \psi_2}= -\frac{1}{\eta}\ln\Big(\sum_{i=1}^2\exp\big(-\eta\rsntt{\psi_i}\big)\Big)
\end{align*}
where $\eta>0$ determines the accuracy of the approximation, i.e., larger values of $\eta$ imply higher accuracy.
\end{assumption}

From now on, when we write $\rho^{\psi}(\boldsymbol{x},t)$, $\rho^{\phi}(\boldsymbol{x},t)$, or $\rho^{\theta}(\boldsymbol{x},t)$ for formulas of class $\psi$, $\phi$, and $\theta$, respectively, we mean the robust semantics including the smooth approximation in Assumption \ref{assumption:2} unless stated otherwise. This approximation is an under-approximation of the robust semantics as remarked in \cite{lindemann2017prescribed}, i.e., the property that $(\boldsymbol{x},t)\models\theta$ if $\rho^{\theta}(\boldsymbol{x},t)>0$ is preserved. The next example illustrates the above and emphasizes that the smooth approximation is only used for conjunctions of \emph{non-temporal formulas} $\psi$.

\begin{example}\label{ex:under_approx}
Assume the formula $\theta:=F_{[5,15]}(\psi_1\wedge\psi_2)\wedge G_{[20,30]}\psi_3$. Then, the robust semantics at time $t:=0$ are $\rho^{\theta}(\boldsymbol{x},0)=\min\big(\max_{t\in[5,15]}(-\frac{1}{\eta}\ln(\exp(-\eta\rsntt{\psi_1})+\exp(-\eta\rsntt{\psi_2})),\min_{t\in[20,30]}\rs{\psi_3}\big)$. 
\end{example}

Second, the next assumption restricts the class of $\psi$ formulas given in \eqref{eq:psi_class} that are contained in \eqref{eq:phi_class} and \eqref{eq:theta_class}.

\begin{assumption}\label{assumption:4}
Each formula of class $\psi$ that is contained in \eqref{eq:phi_class} and \eqref{eq:theta_class} is: 1) such that $\rho^\psi(\boldsymbol{x})$ is concave; and 2) well-posed in the sense that $\rho^\psi(\boldsymbol{x})>0$ implies $\|\boldsymbol{x}\|\le C<\infty$ for some $C\ge 0$.
\end{assumption}

Third, let the global optimum of $\rsnt{\psi}$ be
\begin{align}
\rso{\psi}:=\sup_{\boldsymbol{x}\in\mathbb{R}^n} \rsnt{\psi}
\end{align} 
where $\rsnt{\psi}$ is continuous and concave (Assumption \ref{assumption:2} and \ref{assumption:4}), which simplifies the calculation of $\rso{\psi}$. It holds that $\phi$ is satisfiable, i.e., $\exists\boldsymbol{x}:\mathbb{R}_{\ge 0}\to \mathbb{R}^n$ s.t. $(\boldsymbol{x},0)\models \phi$,  if $\rso{\psi}> 0$.
\begin{assumption}\label{assumption:3}
	The supremum of $\rsnt{\psi}$ is s.t. $\rso{\psi}> 0$. 
\end{assumption} 

 In \cite{lindemann2017prescribed}, we derived a continuous feedback control law to satisfy formulas of class $\phi$ given in \eqref{eq:phi_class}. In this paper, the focus is to derive an event-based feedback control law to satisfy $\phi$. A hybrid control strategy similar to \cite{lindemann2017prescribed} can then be used to satisfy formulas of class $\theta$ given in \eqref{eq:theta_class}. We now summarize the main idea used to achieve $r\le\rsoo{\phi}\le\rho_{max}$, where $r\in\mathbb{R}_{>0}$ is a robustness measure and $\rho_{max}\in\mathbb{R}_{>0}$ with $r<\rho_{max}$ is a robustness delimiter. It then follows that $(\boldsymbol{x},0)\models \phi$ since $r>0$; $r\le\rsoo{\phi}\le\rho_{max}$ is achieved by prescribing a temporal behavior to $\rsntt{\psi}$ through the design parameters $\gamma$ and $\rho_{max}$ as
\begin{align}
	&-\gamma(t)+\rho_{max}< \rsntt{\psi}< \rho_{max}. \label{eq:inequality}
\end{align}
Note the use of $\rsntt{\psi}$ and not $\rsoo{\phi}$ itself. The connection between the non-temporal $\rsntt{\psi}$ and the temporal $\rsoo{\phi}$ is made by the choice of the performance function $\gamma$. The proposed solution in \cite{lindemann2017prescribed} consists of two steps. First, the control law $\boldsymbol{u}$ is designed such that \eqref{eq:inequality} holds for all $t\in\mathbb{R}_{\ge 0}$. In a second step, $\gamma$ is designed such that satisfaction of \eqref{eq:inequality} for all $t\in\mathbb{R}_{\ge 0}$ implies $r\le\rsoo{\phi}\le\rho_{max}$. This second step results in selecting
\begin{align}
&\hspace{0.45cm}t_*\in\begin{cases} a &\text{ if } \phi=G_{[a,b]}\psi \\ 
[a,b] &\text{ if } \phi=F_{[a,b]}\psi,
\end{cases}\label{c_t_star}\\
&\hspace{-0.03cm}\rho_{max}\in\big(\max\big(0,\rho^{\psi}(\boldsymbol{x}(0))\big),\rso{\psi}-\chi\big]\label{rho_max}\\
&\hspace{0.557cm}r\in(0,\rho_{max})\label{r^i}\\
&\hspace{0.3975cm}\gamma_{0}\in\begin{cases}(\rho_{max}-\rho^{\psi}(\boldsymbol{x}(0)),\infty) &\text{if } t_*>0\\
(\rho_{max}-\rho^{\psi}(\boldsymbol{x}(0)),\rho_{max}-r] &\text{if } t_*=0 \end{cases}\label{eq:g1}\\
&\hspace{0.253cm}\gamma_{\infty}\in \Big(0,\min\big(\gamma_0,\rho_{max}-r\big)\Big]\label{eq:g2}\\
&\hspace{0.625cm}l\in \begin{cases}
\mathbb{R}_{\ge 0} & \text{if } -\gamma_0+\rho_{max}\ge r\\
\frac{\ln\big(\frac{r+\gamma_{\infty}-\rho_{max}}{-(\gamma_{0}-\gamma_{\infty})}\big)}{-t_*} & \text{if } -\gamma_0+\rho_{max}< r,\label{eq:g3}
\end{cases}
\end{align}
where $\chi>0$ is a small constant that satisfies $\chi<\rso{\psi}-\max\big(0,\rho^{\psi}(\boldsymbol{x}(0))\big)$. Furthermore, it needs to hold that $\rho^{\psi}(\boldsymbol{x}(0))>r$ if $t_*=0$. This paper will focus on the first step and derive an event-triggered feedback control law such that \eqref{eq:inequality} holds for all $t\in\mathbb{R}_{\ge 0}$. Define the one-dimensional error, the normalized error, and the transformed error as $e(\boldsymbol{x}):=\rsnt{\psi}-\rho_{max}$, $\xi(\boldsymbol{x},t):=\frac{e(\boldsymbol{x})}{\gamma(t)}$, and $\epsilon(\boldsymbol{x},t):=S\big(\xi(\boldsymbol{x},t)\big)=\ln\Big(-\frac{\xi(\boldsymbol{x},t)+1}{\xi(\boldsymbol{x},t)}\Big)$. As a notational rule, when talking about the solution $\boldsymbol{x}$ of \eqref{system_noise} at time $t$, we use $e(t)$, $\xi(t)$, and $\epsilon(t)$, while we use $e(\boldsymbol{x})$, $\xi(\boldsymbol{x},t)$, and $\epsilon(\boldsymbol{x},t)$ when we talk about $\boldsymbol{x}$ as a state; \eqref{eq:inequality} can now be written as $-\gamma(t) <e(t)< 0$, which resembles \eqref{eq:constrained_funnel} by setting $M:=0$ and can further be written as $-1< \xi(t)< 0$. Applying the function $S$ results in $-\infty<\epsilon(t)<\infty$. If now $\epsilon(t)$ is bounded for all $t\ge 0$, then \eqref{eq:inequality} holds for all for all $t\ge 0$. We remark that $\xi\big(\boldsymbol{x}(0),0\big)\in\Omega_\xi:=(-1,0)$ needs to hold initally, which is ensured by the choice of $\gamma_0$. 

\begin{problem}\label{problem}
Consider the system in \eqref{system_noise} and an STL formula $\phi$ of the form \eqref{eq:phi_class}. Design an event-triggered feedback control law $\hat{\boldsymbol{u}}$ s.t. $0< r\le\rsoo{\phi}\le\rho_{max}$, i.e., $(\boldsymbol{x},0)\models \phi$.
\end{problem}}

\section{Control Synthesis}
\label{sec:strategy}
We state the main result upfront in Theorem~\ref{theorem:main} which is proved in the subsequent section. 

 \begin{theorem}[Main Result] \label{theorem:main}
The dynamical system (\ref{system_noise}), satisfying Assumption \ref{assumption:1},
 along with the choice of PPC parameters as per equations \eqref{rho_max} -- \eqref{eq:g3} satisfies an STL formula $\phi$ of the form \eqref{eq:phi_class} if Assumptions \ref{assumption:2} -- \ref{assumption:3} are satisfied and if the event-triggered control law $\hat{\boldsymbol{u}}$ has the form
\begin{equation} \label{eq:ue}
\hat{\boldsymbol{u}}(t):= \boldsymbol{u}(\boldsymbol{x}(t_i),t_i) ~~~~~\forall t \in [t_i,t_{i+1}) 
\end{equation}
where the triggering instances $t_i$ are generated as: 
\begin{align}
t_0&:=0 \nonumber\\
\begin{split}
t_{i+1}&:=\inf\{t>t_i~|~ \|\boldsymbol{x}(t)-\bd x(t_i)\|_\infty> \delta_{i}\\
&\hspace{3.6cm}\text{or } t-t_i>\delta_{i}\}, \label{eq:triggering_rule}
\end{split}
 \hspace{0.5cm}i\ge 1
\end{align}
for some $\delta_{i} > 0$ (obtained later in the paper). The function $\boldsymbol{u}(\boldsymbol{x},t)$ in (\ref{eq:ue}) is chosen as
\begin{align}\label{equ:control}
\boldsymbol{u}(\boldsymbol{x},t):= -\epsilon(\boldsymbol{x},t){g(\boldsymbol{x})}^T\frac{\partial \rsnt{\psi}}{\partial\boldsymbol{x}}.
\end{align}
\end{theorem}


Theorem \ref{theorem:main} is now proved in  three steps. First, Lemma~\ref{theorem:3} summarizes how the continuous feedback control law $\boldsymbol{u}(\boldsymbol{x},t)$ from \cite{lindemann2017prescribed} results in $0< r\le\rsoo{\phi}\le\rho_{max}$. We recall the proof from \cite{lindemann2017prescribed} that is  needed in the third step. Second, we exclude Zeno behavior of the proposed event-triggered control strategy in Lemma \ref{lemma:zeno}. Third, it is  shown in Theorem~\ref{theorem:2} how $\boldsymbol{u}(\boldsymbol{x},t)$ is replaced with the event-triggered control law $\hat{\boldsymbol{u}}(t)$ that still guarantees $0< r\le\rsoo{\phi}\le\rho_{max}$.

\begin{lemma}[Theorem 1 in \cite{lindemann2017prescribed}]\label{theorem:3}
The dynamical system (\ref{system_noise}), satisfying Assumption \ref{assumption:1}, along with the choice of PPC parameters as per equations \eqref{rho_max} -- \eqref{eq:g3} satisfies an STL formula $\phi$ of the form \eqref{eq:phi_class} if the continuous feedback control law \eqref{equ:control} is applied and if Assumptions \ref{assumption:2} -- \ref{assumption:3} are satisfied. It then holds that $0< r\le \rsoo{\phi}\le \rho_{max}$, i.e., $(\boldsymbol{x},0)\models \phi$, with all closed-loop signals being continuous and bounded.

\begin{proof}
First, define the stacked vector $\boldsymbol{y}:=\begin{bmatrix}
\boldsymbol{x}^T & \xi
\end{bmatrix}^T$ and the sets $\Omega_\xi:=(-1,0)$ and $\Omega_{\boldsymbol{x}}:=\{\boldsymbol{x}\in\mathbb{R}^n|-1<\xi(\boldsymbol{x},0):=\frac{\rsnt{\psi}-\rho_{max}}{\gamma_0}<0\}$ and note that $\xi(0)\in\Omega_\xi$ and $\boldsymbol{x}(0)\in\Omega_{\boldsymbol{x}}$ due to the choice of $\gamma_0$. As in \cite{lindemann2017prescribed}, it follows that the conditions in \cite[Theorem 54]{sontag2013mathematical} hold. Consequently, there exists a maximal solution $\boldsymbol{y}:\mathcal{J}\to\Omega_{\boldsymbol{y}}$ with $\mathcal{J}:=[0,\tau_{max})$ and $\tau_{max}>0$, i.e., $\xi(t)\in\Omega_\xi$ and $\boldsymbol{x}(t)\in\Omega_{\boldsymbol{x}}$ for all $t \in \mathcal{J}$. 

We next show that $\boldsymbol{y}$ is complete, i.e.,  $\tau_{max}=\infty$, by contradiction of \cite[Proposition C.3.6]{sontag2013mathematical}. Assume therefore $\tau_{max}<\infty$ and consider the Lyapunov function $V(\epsilon):=\frac{1}{2}\epsilon^2$ and define $\dot{V}:=\frac{\partial V}{\partial \epsilon}\frac{d \epsilon}{dt}$. Thus, it holds that 
\begin{align}\label{eq:p1}
\dot{V}=\epsilon\dot{\epsilon} =\epsilon\Big(-\frac{1}{\gamma\xi(1+\xi)}\big(\frac{\partial \rsnt{\psi}}{\partial \boldsymbol{x}}^T\dot{\boldsymbol{x}}-\xi\dot{\gamma}\big)\Big).
\end{align} 
Define $\alpha:=-\frac{1}{\gamma\xi(1+\xi)}$, which satisfies $\alpha(t)\in[\frac{4}{\gamma_0},\infty)\subset \mathbb{R}_{> 0}$ for all $t\in \mathcal{J}$. Inserting \eqref{system_noise} and \eqref{equ:control} into \eqref{eq:p1} results in 
\begin{align*}
\hspace{-0.05cm}\dot V= \alpha\epsilon\Big(\frac{\partial \rsnt{\psi}}{\partial \boldsymbol{x}}^T\big(f(\boldsymbol{x})-\epsilon g(\boldsymbol{x}){g(\boldsymbol{x})}^T\frac{\partial \rsnt{\psi}}{\partial \boldsymbol{x}}+ \boldsymbol{w}\big)-\xi\dot{\gamma}  \Big)
\end{align*}
which can now be upper bounded as
\begin{align*}
\dot V\le \alpha|\epsilon|\Big(k_1- |\epsilon| \lambda_{\min}\|\frac{\partial \rsnt{\psi}}{\partial \boldsymbol{x}}\|^2\Big) \le \alpha |\epsilon|(k_1-k_2|\epsilon|)  
\end{align*}
where the positive constants $k_1$ and $k_2$ can be obtained as follows. For the constant $k_1$ note that $\boldsymbol{w}$, $\xi$, and $\dot{\gamma}$ are bounded and that continuous functions on compact domains are bounded.  Note especially that $\frac{\partial \rsnt{\psi}}{\partial \boldsymbol{x}}$ is continuous on the compact set $\text{cl}(\Omega_{\boldsymbol{x}})$ where $\text{cl}$ denotes the set closure. For $k_2$, a positive lower bound for $\|\frac{\partial \rsnt{\psi}}{ \partial \bd x}\|$ can be derived. Since $\rsnt{\psi}$ is a smooth and concave function due to Assumption~\ref{assumption:2} and \ref{assumption:4}, we have $\|\frac{\partial \rsnt{\psi}}{ \partial \bd x}\|\ge \frac{\rho^{\psi}_{opt}-\rsnt{\psi}}{\|\bd x^* - \bd x\|}$ where $\rho^\psi(\bd x^*)=\rho^\psi_{opt}$. It holds that $\rho_{max}\le \rho^\psi_{opt}-\chi<\rho^\psi_{opt}$ due to \eqref{rho_max}, which leads to $\rsntt{\psi}<\rho_{max}\le \rho^\psi_{opt}-\chi<\rho^\psi_{opt}$ since \eqref{eq:inequality} holds for all $t\in \mathcal{J}$. Hence, there exists $\kappa_1$ with $ \chi\ge \kappa_1 >0$ such that $\kappa_1\le \rho^{\psi}_{opt}-\rsntt{\psi}$. Furthermore, $\|\bd x^*-\bd x(t)\|$ is upper bounded since  $\bd x^*$ is finite and since $\bd x(t)\in\Omega_{\boldsymbol{x}}$ ($\Omega_{\boldsymbol{x}}$ is bounded) for all $t\in\mathcal{J}$ so that there exists a $\kappa_2>0$ such that $\|\bd x^*-\bd x(t)\|\le\kappa_2$. Thus, $\|\frac{\partial \rsnt{\psi}}{ \partial \bd x}\|\ge \frac{\kappa_1}{\kappa_2}>0$ and we set $k_2=\lambda_{\text{min}}(\frac{\kappa_1}{\kappa_2})^2$.  It follows that $\dot{V}\le 0$ if $\frac{k_1}{k_2}\le|\epsilon|$ and it can be concluded that the transformed error $|\epsilon|$ will be upper bounded due to the level sets of $V(\epsilon)$ as $|\epsilon(t)|\le \max\left(|\epsilon(0)|,\frac{k_1}{k_2}\right)$, i.e., $\epsilon(t)$ is lower and upper bounded and hence evolves in a compact set.  By the same arguments as in \cite{lindemann2017prescribed} it follows that there exists compact sets $\Omega_{\xi}^\prime\subset \Omega_{\xi}$ and $\Omega_{\boldsymbol{x}}^\prime\subset \Omega_{\boldsymbol{x}}$ such that $\xi(t)\in \Omega_{\xi}^\prime$ and $\boldsymbol{x}(t)\in\Omega_{\boldsymbol{x}}^\prime$ for all $t\in\mathcal{J}$. According to \cite[Proposition C.3.6]{sontag2013mathematical} it follows by contradiction that $\tau=\infty$. By the choice of $\gamma$ and \cite[Theorem~2]{lindemann2017prescribed} it holds that  $0< r\le \rsoo{\phi}\le \rho_{max}$.
\end{proof}
\end{lemma}


As an intermediate step, we next show that the triggerings generated by the rule \eqref{eq:triggering_rule} do not exhibit Zeno behavior.
\begin{lemma}\label{lemma:zeno}
The event-triggered control law $\hat{\boldsymbol{u}}(t)$ in \eqref{eq:ue} in conjunction with the triggering rule \eqref{eq:triggering_rule} does not exhibit Zeno behavior, i.e., $t_{i+1}-t_i$ is lower bounded for all $i\in\mathbb{N}$.

\begin{proof}
Triggerings induced when $\inf\{t|t-t_i> \delta_{i}\}$ imply that $t_{i+1}-t_i\ge \delta_{i}$. Otherwise, i.e., $\|\boldsymbol{x}(t)-\bd x(t_i)\|_\infty> \delta_{i}$, we have $\bd x(t)=\bd x(t_i)+\int_{t_i}^t[f(\bd x(s))+g(\bd x(s))\hat{\bd u }(s)+\bd w(s)]ds$ which is equivalent to $\bd x(t)-\bd x(t_i)=\int_{t_i}^tK_1(s)ds+\int_{t_i}^t[f(\bd x(s))-f(\bd x(t_i))]ds+\int_{t_i}^t[g(\bd x(s))-g(\bd x(t_i))]\hat{\bd u }(s)ds$ where $K_1(s):=f(\bd x(t_i))+g(\bd x(t_i))\hat{\bd u}(s)+\bd w(s)$. Then
\begin{small}
\begin{align*}\|\bd x(t)-\bd x(t_i)\|_\infty &\le \int_{t_i}^tK(s)ds+\int_{t_i}^tL_0(s)\|\bd x(s)-\bd x(t_i)\|_\infty ds
\end{align*} 
\end{small}where $K(s):=\|K_1(s)\|_\infty$ and $L_0(s):=\|L_f+L_g\hat{\bd u}(s)\|_\infty$ with $L_f$ and $L_g$ being the Lipschitz constants of the functions $f(\bd x)$ and $g(\bd x)$ in the domain $B_{\delta_{i}}(\bd x(t_i))$. Thus using the Gr\"{o}nwall-Bellman inequality, it can be shown that
\begin{align*}
\|\bd x(t)-\bd x(t_i)\|_\infty \le \Big(\int_{t_i}^tK(s)ds\Big)e^{\int_{t_i}^tL_0(s)ds}.
\end{align*}
In order for $\bd x(t)$ to leave the domain $B_{\delta_{i}}(\bd x(t_i))$, which corresponds to the condition in \eqref{eq:triggering_rule}, it is necessary that $\zeta(t):=\Big(\int_{t_i}^tK(s)ds\Big)e^{\int_{t_i}^tL_0(s)ds} \ge \delta_{i}$. Clearly $\zeta(t_i)=0$ and $\zeta(t)$ is differentiable everywhere for all $t>t_i$ with finite $\dot{\zeta}(t)$ and hence $\zeta(t)$ is Lipschitz continuous. Let us denote its Lipschitz constant by $L_{\zeta}$. Therefore at the next triggering instance $t_{i+1}$ we have $L_{\zeta}(t_{i+1}-t_i)\ge \zeta(t_{i+1}) \ge \delta_{i}$ so that $t_{i+1}-t_i\ge \frac{\delta_{i}}{L_{\zeta}}$ and
hence Zeno behavior is excluded.
\end{proof}
\end{lemma} 

The focus of this work is to design an event-based control law $\hat{\bd u}(t)$ based on the continuous feedback control law $\boldsymbol{u}(\boldsymbol{x},t)$ in \eqref{equ:control}. We show that simply replacing  \eqref{equ:control} by its equivalent zero-order hold approximation will be sufficient.  

\begin{theorem}\label{theorem:2}
With the same assumptions as in Lemma~\ref{theorem:3}, the event-triggered control law $\hat{\boldsymbol{u}}(t)$ in \eqref{eq:ue} in conjunction with the triggering rule \eqref{eq:triggering_rule} guarantees $0 <r\le \rsoo{\phi}\le\rho_{max}$, i.e., $(\boldsymbol{x},0)\models \phi$, provided that $\|\boldsymbol{u}(\boldsymbol{x},t)-\hat{\boldsymbol{u}}(t)\|_\infty\le \delta_{\boldsymbol{u}}$ for all $t\in \mathbb{R}_{\ge 0}$, where $\delta_{\boldsymbol{u}}>0$ is a design parameter.
\end{theorem}

\begin{proof}
We now need to show that $\epsilon(t)$ is bounded between each of the triggering instances $t_i$ and $t_{i+1}$. We can similarly to Lemma \ref{theorem:3} guarantee a maximal solution $\boldsymbol{y}:\mathcal{J}\to\Omega_{\boldsymbol{y}}$ with $\mathcal{J}:=[t_i,\tau_{max})$. Consider again the Lyapunov function $V(\epsilon):=\frac{1}{2}\epsilon^2$ so that (recall \eqref{eq:p1})
\begin{align*}
\dot{V}&=\alpha\epsilon\Big(\frac{\partial \rsnt{\psi}}{\partial \boldsymbol{x}}^T\big(f(\boldsymbol{x})+g(\boldsymbol{x})\hat{\boldsymbol{u}}+\boldsymbol{w}\big)-\xi\dot{\gamma}\Big)\\
&= \alpha\epsilon \Big(\frac{\partial \rsnt{\psi}}{\partial \boldsymbol{x}}^T\big(f(\boldsymbol{x})+g(\boldsymbol{x})(\hat{\boldsymbol{u}}+\boldsymbol{u}-\boldsymbol{u})+\boldsymbol{w}\big)-\xi\dot{\gamma}\Big)  \\
&\le \alpha|\epsilon| k_1- \alpha\epsilon^2 k_2+\alpha|\epsilon| \Big\|\frac{\partial \rsnt{\psi}}{\partial \boldsymbol{x}}^Tg(\boldsymbol{x})\Big\|_\infty\|\hat{\boldsymbol{u}}-\boldsymbol{u}\|_\infty,
\end{align*}
where $k_1$ and $k_2$ are from the proof of Lemma \ref{theorem:3}. We can write $\|\frac{\partial \rsnt{\psi}}{\partial \boldsymbol{x}}^Tg(\boldsymbol{x})\|_\infty\le k_3$ for some positive constant $k_3$ and obtain finally $\dot V\le   \alpha|\epsilon|( k_1 -\epsilon  k_2 +k_3 \|\boldsymbol{u}-\hat{\boldsymbol{u}}\|_\infty).$ Thus, $\|\boldsymbol{u}-\hat{\boldsymbol{u}}\|_\infty\le \delta_{\boldsymbol{u}}$ implies that $\|\epsilon(t)\|\le \max\{\epsilon(0), \frac{k_1+k_3\delta_{\boldsymbol{u}}}{k_2}\}$. Together with Lemma \ref{lemma:zeno}, it can be concluded that $\tau_{max}=t_{i+1}$ and hence  $0< r\le\rsoo{\phi}\le\rho_{max}$.
\end{proof}

According to Theorem \ref{theorem:2} there is no bound imposed on the value of $\delta_{\boldsymbol{u}}$, i.e., $\delta_{\boldsymbol{u}}$ is a design parameter. Larger values of $\delta_{\boldsymbol{u}}$ imply larger inter-event times, whereas smaller values of $\delta_{\boldsymbol{u}}$ imply more frequent triggering of the events. However, larger values of $\delta_{\boldsymbol{u}}$ also imply that $\epsilon(\bd x,t)$ can attain higher values, which implies (according to \eqref{equ:control}) larger magnitude for the control signal. Theorem \ref{theorem:2} may not be very useful in practice since $\bd u(\bd x,t)$ still needs to be computed continuously at the sensors in order to ensure $\|\boldsymbol{u}(\boldsymbol{x},t)-\hat{\boldsymbol{u}}(t)\|_\infty\le \delta_{\boldsymbol{u}}$. The triggering rule in \eqref{eq:triggering_rule}, derived in the sequel, avoids this and is chosen in a way to ensure $\|\boldsymbol{u}(\boldsymbol{x},t)-\hat{\boldsymbol{u}}(t)\|_\infty\le \delta_{\boldsymbol{u}}$. {It holds that $\boldsymbol{u}(\boldsymbol{x},t)$ is a Lipschitz continuous function on $B_{\delta_{i,\boldsymbol{x}}}(\boldsymbol{x}(t_i))\times[t_i,t_i+\delta_{i,t}]$ for some $\delta_{i,\boldsymbol{x}},\delta_{i,t}>0$. To see this, note that $\Omega:=\{(\boldsymbol{x},t)\in\mathbb{R}^n\times\mathbb{R}_{\ge 0}|-1<\xi(\boldsymbol{x},t)<0\}$ is an open set and that $\epsilon(\boldsymbol{x},t)$, $g(\boldsymbol{x})$, and $\frac{\partial \rsnt{\psi}}{\partial\boldsymbol{x}}$ are locally Lipschitz continuous on $\Omega$ due to being continuously differentiable on $\Omega$. If now $(\boldsymbol{x}(t_i),t_i)\in\Omega$, then there exists a set $B_{\delta_i,\boldsymbol{x}}(\boldsymbol{x}(t_i))\times[t_i,t_i+\delta_{i,t}]\subset\Omega$  in which $\epsilon(\boldsymbol{x},t)$, $g(\boldsymbol{x})$, $\frac{\partial \rsnt{\psi}}{\partial\boldsymbol{x}}$, and hence $\boldsymbol{u}(\boldsymbol{x},t)$ are Lipschitz continuous. Define $\boldsymbol{z}(t):=\begin{bmatrix} {\boldsymbol{x}(t)}^T & t \end{bmatrix}^T$ and denote by $L_{\boldsymbol{z}}(\delta_{i,\boldsymbol{x}},\delta_{i,t})$ the Lipschitz constant of $\bd u(\boldsymbol{x},t)$ on $B_{\delta_{i,\boldsymbol{x}}}(\boldsymbol{x}(t_i))\times[t_i,t_i+\delta_{i,t}]$, i.e., $\|\bd u(\boldsymbol{x}(t_1),t_1)-\bd u(\boldsymbol{x}(t_2),t_2)\|_\infty\le L_{\boldsymbol{z}}(\delta_{i,\boldsymbol{x}},\delta_{i,t})\|\boldsymbol{z}(t_1)-\boldsymbol{z}(t_2)\|_\infty$ for $\boldsymbol{x}(t_1),\boldsymbol{x}(t_2)\in B_{\delta_{i,\boldsymbol{x}}}(\boldsymbol{x}(t_i))$ and $t_1,t_2\in[t_i,t_i+\delta_{i,t}]$. Note that $L_{\boldsymbol{z}}(\delta_{i,\boldsymbol{x}},\delta_{i,t})$ is a non-decreasing function of $\delta_{i,\boldsymbol{x}}$ and $\delta_{i,t}$.} Let us consider 
\begin{align}\label{eq:delta_x}
\delta_{i}:=\min\Big(\frac{\delta_u}{L_{\boldsymbol{z}}(\delta_{i,\boldsymbol{x}},\delta_{i,t})},\delta_{i,\boldsymbol{x}},\delta_{i,t}\Big),
\end{align} 
which implies $B_{\delta_i}(\bd x(t_i))\subseteq B_{\delta_{i,\boldsymbol{x}}}(\bd x(t_i))$ and $[t_i,t_i+\delta_i]\subseteq[t_i,t_i+\delta_{i,t}]$. Thus, for all $t\in[t_i,t_i+\delta_i]$ and $\bd x(t) \in B_{\delta_{i}}(\bd x(t_i))$, i.e., when $\|\bd z(t)-\bd z(t_i)\|_\infty\le \delta_i$, it holds that
$$\|\bd{u}(\bd{x}(t),t)-\hat {\bd u}(t)\|_\infty \le L_{\boldsymbol{z}}(\delta_{i,\boldsymbol{x}},\delta_{i,t}) \|\bd z(t)-\bd z(t_i)\|_\infty\le \delta_{\bd u}.$$
Due to the use of the $\|\cdot\|_\infty$-norm, $\bd x(t) \in B_{\delta_{i}}(\bd x(t_i))$ and $t\in[t_i,t_i+\delta_i]$ is a sufficient condition to ensure $\|\boldsymbol{u}(\boldsymbol{x},t)-\hat{\boldsymbol{u}}(t)\|_\infty\le \delta_{\boldsymbol{u}}$. Thus, the $(i+1)$-th triggering instance is induced when
\begin{align*} 
t_{i+1}=\inf \{t>t_i|\|\bd x(t)-\bd x(t_i)\|_\infty>{\delta_{i}} \text{ or } t-t_i>\delta_i\},
\end{align*}
which is the triggering rule given in \eqref{eq:triggering_rule}.
\begin{remark}
With the choice of $\delta_{i}$ in \eqref{eq:delta_x},  $\|\bd x(t)-\bd x(t_i)\|_\infty\le \delta_{i}$ and $t-t_i\le \delta_i$ is a sufficient condition for $\|\boldsymbol{u}(\boldsymbol{x},t)-\hat{\boldsymbol{u}}(t)\|\le \delta_{\boldsymbol{u}}$. Hence 
{Theorem \ref{theorem:main} is finally obtained by the choice of $\delta_{i}$ in conjunction with Theorem \ref{theorem:2}.}
\end{remark}


\section{Simulations}
\label{sec:simulations}

\begin{figure*}[tbh]
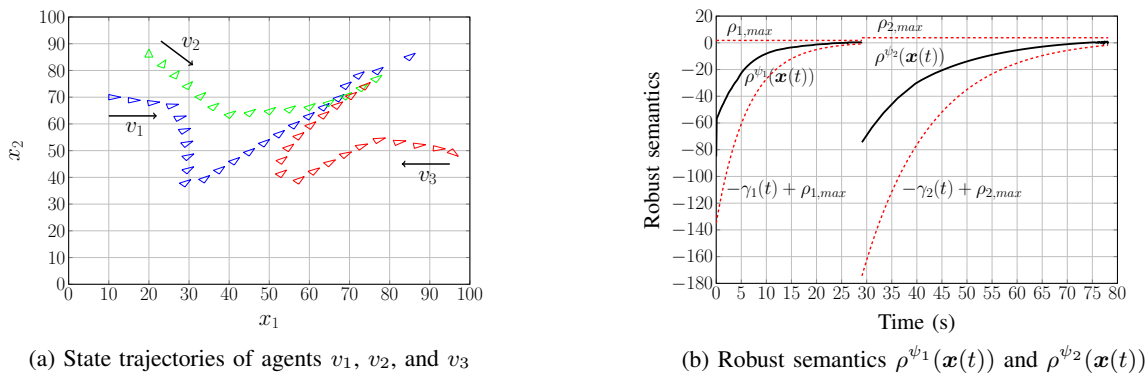

\centering
\begin{subfigure}{0.49\textwidth}
\input{figures/1}\caption{State trajectories of agents $v_1$, $v_2$, and $v_3$}\label{fig:1}
\end{subfigure}
\begin{subfigure}{0.49\textwidth}
\input{figures/2}\caption{Robust semantics $\rho^{\psi_1}(\boldsymbol{x}(t))$ and $\rho^{\psi_2}(\boldsymbol{x}(t))$}\label{fig:2}
\end{subfigure}
\caption{Agent Trajectories and Robust Semantics}
\vspace{-10 pt}
\end{figure*}

\begin{figure}[h]
\centering
\input{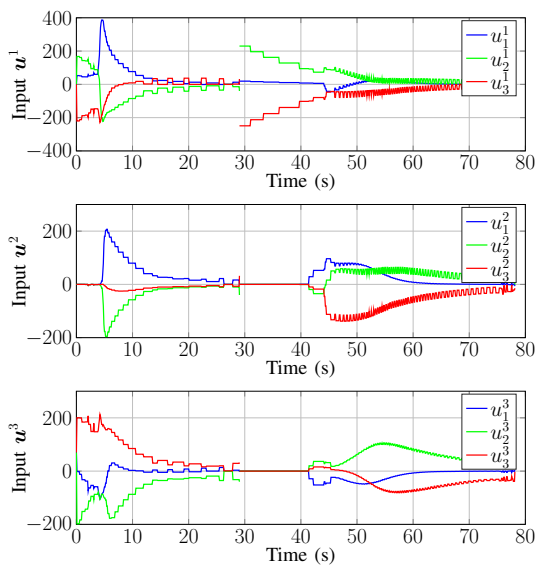}
\caption{Input trajectories of agents $v_1$, $v_2$, and $v_3$}
\label{fig:3}
\vspace{-10 pt}
\end{figure}

We consider a centralized multi-agent system consisting of three agents $v_1$, $v_2$, and $v_3$. Each agent is a three-wheeled omni-directional mobile robot as in \cite{liu2008omni} with three states: two states $x_1$ and $x_2$ describing the robot's position and one state $x_{3}$ describing its orientation with respect to the $x_{1}$-axis. In figures, the orientation will be indicated by triangles. The states of each agent $v_i$ with $i\in\{1,2,3\}$ are hence described by $\boldsymbol{x}^i:=\begin{bmatrix} x_1^i & x_2^i & x_3^i \end{bmatrix}$, while the control input is $\boldsymbol{u}^i:=\begin{bmatrix}u_1^i & u_2^i & u_3^i\end{bmatrix}$. The dynamics of each robot are 
\begin{align*}
\dot{\boldsymbol{x}}^i=g_i(\boldsymbol{x}^i)\boldsymbol{u}^i=\begin{bmatrix}
\cos(x_{3}^i) & -\sin(x_{3}^i) & 0\\
\sin(x_{3}^i) & \cos(x_{3}^i) & 0\\
0 & 0 & 1
\end{bmatrix}\Big(B_i^T\Big)^{-1}R_i\boldsymbol{u}^i,
\vspace*{-3pt}
\end{align*}
where $B_i:=\begin{bmatrix}
0 & \cos(\pi/6) & -\cos(\pi/6)\\
-1 & \sin(\pi/6) & \sin(\pi/6)\\
L_i & L_i & L_i
\end{bmatrix}$ describes geometrical constraints with $L_i:=0.2$ and $R_i=0.02$ as the radius of the robot body and the wheels, respectively. By definining $\boldsymbol{x}:=\begin{bmatrix}
\boldsymbol{x}^1 & \boldsymbol{x}^2 & \boldsymbol{x}^3
\end{bmatrix}^T$ and $\boldsymbol{u}:=\begin{bmatrix}
\boldsymbol{u}^1 & \boldsymbol{u}^2 & \boldsymbol{u}^3
\end{bmatrix}^T$, the overall dynamics are then given by
\begin{align*}
\dot{\boldsymbol{x}}:=\text{diag}\big(g_1(\boldsymbol{x}^1),g_2(\boldsymbol{x}^2),g_3(\boldsymbol{x}^3)\big)\boldsymbol{u}=g(\boldsymbol{x})\boldsymbol{u}.
\end{align*}  

The STL task imposed on the multi-agent system is a formula $\theta:=F_{[0,50]}\psi_1 \wedge F_{[50,100]} \psi_2$ where $\psi_1$ orders agent $v_1$, $v_2$, and $v_3$ to the positions $\begin{bmatrix}
20 & 30
\end{bmatrix}^T$, $\begin{bmatrix}
40 & 60
\end{bmatrix}^T$, and $\begin{bmatrix}
60 & 30
\end{bmatrix}^T$, respectively, while eventually all agents have an orientation of $45$ degrees. Furthermore, agent $v_1$ and $v_3$ should stay close; $\psi_2$ orders agent $v_1$ to $\begin{bmatrix}
90 & 90
\end{bmatrix}^T$, while agent $v_1$ and $v_2$ and agent $v_2$ and $v_3$ should stay in proximity and while all agents remain the orientation of $45$ degrees. In formulas, this can be expressed as $\psi_1:=(\|\begin{bmatrix}
x_1^1 & x_2^1
\end{bmatrix}^T-\begin{bmatrix}
20 & 30
\end{bmatrix}^T\|<10)\wedge (\|\begin{bmatrix}
x_1^2 & x_2^2
\end{bmatrix}^T-\begin{bmatrix}
40 & 60
\end{bmatrix}^T\|<10)\wedge (\|\begin{bmatrix}
x_1^3 & x_2^3
\end{bmatrix}^T-\begin{bmatrix}
60 & 30
\end{bmatrix}^T\|<10)\wedge (\|\begin{bmatrix}
x_1^1 & x_2^1
\end{bmatrix}^T-\begin{bmatrix}
x_1^3 & x_2^3
\end{bmatrix}^T\|<30)\wedge(|x_3^1-45|<5)\wedge(|x_3^2-45|<5)\wedge(|x_3^3-45|<5)$ and $\psi^2:=(\|\begin{bmatrix}
x_1^1 & x_2^1
\end{bmatrix}^T-\begin{bmatrix}
90 & 90
\end{bmatrix}^T\|<10)\wedge(\|\begin{bmatrix}
x_1^1 & x_2^1
\end{bmatrix}^T-\begin{bmatrix}
x_1^2 & x_2^2
\end{bmatrix}^T\|<10)\wedge(\|\begin{bmatrix}
x_1^2 & x_2^2
\end{bmatrix}^T-\begin{bmatrix}
x_1^3 & x_2^3
\end{bmatrix}^T\|<10)\wedge(|x_3^1-45|<5)\wedge(|x_3^2-45|<5)\wedge(|x_3^3-45|<5)$. With $\eta:=1$, it holds that $\rho^{\psi_1}_{opt}=1.86$ and $\rho^{\psi_2}_{opt}=3.89$ so that $\rho_{max}^{\psi_1}:=1.8$, $\rho_{max}^{\psi_1}:=3.8$, $r_1:=0.5$, $r_2:=1$, and $\delta_{\boldsymbol{u}}:=50$ have been selected.

All simulations have been performed in real-time on a two-core 1,8 GHz CPU with 4 GB of RAM.
 The agent trajectories in the $x_1$-$x_2$ plane are displayed in Fig. \ref{fig:1}, while the funnels \eqref{eq:inequality}, including $\rho^{\psi_1}(\boldsymbol{x}(t))$ and $\rho^{\psi_2}(\boldsymbol{x}(t))$, are shown in Fig. \ref{fig:2}. The formula $\theta$ is satisfied, i.e, $(\boldsymbol{x},0)\models \theta$, and it holds that $\min(r_1,r_2)=0.5<\rho^{\theta}(\boldsymbol{x},0)<1.8=\min(\rho_{1,max},\rho_{2,max})$. The control inputs are shown in Fig.~\ref{fig:3} and it is visible that during the satisfaction of the first subformula $F_{[0,50]}\psi_1$ there are fewer control updates than for the second subformula $F_{[50,100]} \psi_2$. The coordination of agent $v_1$, $v_2$, and $v_3$ leads to an increase in control updates in the latter case. {The sampling frequency has been set to $100$ Hz; and total simulation duration was $77.25$ seconds. Out of $7725$ samples, our event-triggered control law only required $185$ triggerings. Thereby, a reduction in communication and computation events by $97.6$~\% has been achieved.}

\section{Conclusion}
\label{sec:conclusion}
In this paper, we have derived an event-triggered feedback control law for dynamical systems under signal temporal logic tasks. The event-triggering mechanism is based on a norm bound on the difference between the  continuous feedback law and the event-triggered version of it. Event-triggered control leads to a significant decrease in communication between the sensors and the actuators, which is highly desirable under costly communication. 


\bibliographystyle{IEEEtran}
\bibliography{literature}

\addtolength{\textheight}{-12cm}   

\end{document}